\documentclass[conference]{IEEEtran}
\IEEEoverridecommandlockouts
\usepackage{cite}
\usepackage{amsmath,amssymb,amsfonts,amsthm}
\usepackage{graphicx}
\usepackage{textcomp}
\usepackage{xcolor}
\usepackage{multirow}
\usepackage{mathrsfs}
\usepackage{physics}
\usepackage[final]{microtype}
\usepackage[hidelinks]{hyperref}
\usepackage{nicefrac}
\usepackage{tikz}

\newcommand\copyrighttext{%
	\footnotesize \textcopyright 2024 IEEE. Personal use of this material is permitted.
	Permission from IEEE must be obtained for all other uses, in any current or future
	media, including reprinting/republishing this material for advertising or promotional
	purposes, creating new collective works, for resale or redistribution to servers or
	lists, or reuse of any copyrighted component of this work in other works.
	DOI: \href{https://doi.org/10.1109/LCSYS.2024.3406053}{10.1109/LCSYS.2024.3406053}}
\newcommand\copyrightnotice{%
	\begin{tikzpicture}[remember picture,overlay]
		\node[anchor=south,yshift=10pt] at (current page.south) {\fbox{\parbox{\dimexpr\textwidth-\fboxsep-\fboxrule\relax}{\copyrighttext}}};
	\end{tikzpicture}%
}

\newtheorem{proposition}{Proposition}
\newtheorem{lemma}{Lemma}
\newtheorem{definition}{Definition}
\newtheorem{remark}{Remark}

\DeclareMathOperator{\vstack}{col}

\def\BibTeX{{\rm B\kern-.05em{\sc i\kern-.025em b}\kern-.08em
    T\kern-.1667em\lower.7ex\hbox{E}\kern-.125emX}}

\begin{document}
\title{Exploring the Links between the
Fundamental Lemma and Kernel Regression}
\author{Oleksii Molodchyk\href{https://orcid.org/0009-0001-1659-1891}{\includegraphics[scale=0.09]{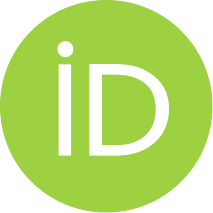}}, Timm Faulwasser\href{https://orcid.org/0000-0002-6892-7406}{\includegraphics[scale=0.09]{orcid.pdf}}, \IEEEmembership{Senior Member, IEEE}
\thanks{OM and TF are with the Institute of Energy Systems, Energy Efficiency and Energy Economics, TU Dortmund University, 44227 Dortmund, Germany (e-mails: oleksii.molodchyk@tuhh.de; timm.faulwasser@ieee.org). TF is also with the Institute of Control Systems, TU Hamburg, 20173 Hamburg, Germany.}}

\maketitle
\copyrightnotice

\begin{abstract}
    Generalizations and variations of the fundamental lemma by Willems \textit{et al.} are an active topic of recent research. In this note, we explore and formalize the links between kernel regression and some known nonlinear extensions of the fundamental lemma. Applying a transformation to the usual linear equation in Hankel matrices, we arrive at an alternative implicit kernel representation of the system trajectories while keeping the requirements on persistency of excitation. We show that this representation is equivalent to the solution of a specific kernel regression problem. We explore the possible structures of the underlying kernel as well as the system classes to which they correspond.
\end{abstract}


\section{Introduction}
\label{sec:introduction}
Any dynamic system can be associated with its behavior, i.e., with the set of all input-output-state trajectories it can generate. Direct data-driven methods exploit this to bypass the identification of model parameters and the usage of parametric state-space models. For finite-dimensional linear time-invariant (LTI) systems, a well-known data-driven description is due to Jan C. Willems and co-authors; it is known as the \emph{fundamental lemma}~\cite{Willems2005}. Its key insight is that pre-recorded input-output data spans the space of all finite-length input-output trajectories whereby the pre-recorded input sequence has to be persistently exciting and the underlying system is controllable. The fundamental lemma is at the core of various data-driven control methods, cf. the overviews~\cite{Markovsky2021, Verheijen2023} and references therein. Yet, the original formulation of the lemma is limited by the LTI assumption imposed on the system dynamics.

Hence, a lot of recent research is focused on generalizing the fundamental lemma towards data-driven descriptions of nonlinear systems. One approach is to consider nonlinear systems that admit a decomposition of the dynamics into a linear combination of nonlinear basis functions. This allows to wrap the behavior of a nonlinear system into LTI dynamics such that linear approaches become applicable, cf. the notion of \emph{LTI embedding} \cite{Markovsky2023}.

In this context, there exist several approaches and results to relax the linearity assumptions of the fundamental lemma. For instance, \cite{Berberich2020} considers multiple-input multiple-output (MIMO) Hammerstein-Wiener systems;~\cite{Alsalti2021} discusses single-input single-output (SISO) flat systems and the induced feedback linearization properties; for a recent expansion to the MIMO case, see \cite{Alsalti2023}. Alternatively, \cite{Markovsky2023} leverages basis functions to unify Hammerstein, Volterra, and bilinear systems in one framework. Other works tailoring the fundamental lemma to nonlinear systems include \cite{Rueda-Escobedo2020, Yuan2022, Reinhardt2023, Lazar2023}.

Another avenue towards data-driven description of nonlinear systems considers reproducing kernel Hilbert spaces (RKHS)~\cite{Lian2021, Huang2023}. The underlying advantage is that any RKHS can be uniquely characterized by a single positive definite kernel function \cite{Aronszajn1950, Saitoh2016}. This leads to concise formulations that can be used for data-driven predictive control of nonlinear systems. Similar to regression in RKHS, Gaussian Processes (GPs) can provide data-based predictions by interpreting kernel functions as measures of similarity. This renders GPs suitable for data-driven predictive control applications~\cite{Umlauft2020,Maiworm2021}.

One of the drawbacks of kernel formulations is that the deduction of the requirements on the input signals (e.g., persistency of excitation) is not straightforward. Hence, in this letter, we investigate the correspondence between specific kernel structures and data-driven modeling of specific classes of nonlinear systems. In particular, we use the implicit kernel formulation from \cite{Huang2023} and we show which kernel structures allow us to deduce the required order of persistency of excitation of the inputs. In contrast to \cite{Lazar2023}, which also considers kernel regression for nonlinear systems, we derive a result that stays equivalent to the original fundamental lemma \cite{Willems2005} in case a linear kernel is selected. 

The remainder of the paper is structured as follows: Section~\ref{sec:preliminaries} recalls the fundamental lemma and discusses kernel regression and its RKHS interpretation. Later, we recap the kernelized setting of the fundamental lemma from \cite{Huang2023}, which we use as a blueprint. In Section~\ref{sec:results}, we discuss the kernel structures that can be employed for nonlinear data-driven prediction as well as the associated identifiability requirements. The letter ends with conclusions in Section~\ref{sec:conclusion}. 
\vspace*{2mm}

\subsubsection*{Notation}
For $w : \mathbb{N}_0 \to \mathbb{R}^{n_w}$ and  $a,b \in \mathbb{N}_0: b \geq a$, the restriction of $w$ to $\mathbb{I}_{[a,b]} \doteq \left[a,b\right] \cap \mathbb{N}_0$ is denoted as $w_{[a,b]}$. We identify $T \doteq b - a + 1$ as the length of $w_{[a,b]}$. Whenever no confusion can arise, we interpret $w_{[a,b]}$ as the column-stacked vector $w_{[a,b]} \doteq \vstack(w(a), \ldots, w(b)) \in \mathbb{R}^{Tn_w}$. The function $\vstack$ stacks a list of column vectors/matrices vertically, preserving its original order. We denote the identity matrix as $I$, the Frobenius norm as $\norm{\cdot}_F$, and the Euclidean norm as $\norm{\cdot}_2$.

\section{Preliminaries} 
\label{sec:preliminaries} 
\subsection{The Fundamental Lemma}
For $w_{[0, T-1]}$ with $T \in \mathbb{N}$ let the Hankel matrix of depth $L \in \mathbb{I}_{[1,T]}$ be defined as
\[
    \mathscr{H}_L(w_{[0, T-1]}) \doteq \mqty[w_{[0,L-1]} & w_{[1,L]} & \cdots & w_{[T-L,T-1]}].
\]
\begin{definition}[Persistency of excitation]
    A given $T$-long trajectory $w_{[0, T-1]}$ is called persistently exciting of order $L$ if $\rank \mathscr{H}_L(w_{[0, T-1]}) = L n_w$. \hfill {\small{$\Box$}}
\end{definition}
\vspace*{1mm}

Consider a discrete-time LTI system in minimal state-space representation
\begin{equation}
	\label{eq:lti_system}
	\left\{
	\begin{array}{l l}
		x(t+1) &= A x(t) + B u(t), \quad x(0) \doteq x^0,\\
		y(t) &= C x(t) + D u(t). 
	\end{array}
	\right.
\end{equation}
Here, the state trajectory $x: \mathbb{N}_0 \to \mathbb{R}^{n_x}$, the input trajectory $u: \mathbb{N}_0 \to \mathbb{R}^{n_u}$, and the output trajectory $y: \mathbb{N}_0 \to \mathbb{R}^{n_y}$ are related through the matrices $A \in \mathbb{R}^{n_x \times n_x}$, $B \in \mathbb{R}^{n_x \times n_u}$, $C \in \mathbb{R}^{n_y \times n_x}$, and $D \in \mathbb{R}^{n_y \times n_u}$ at each time $t \in \mathbb{N}_0$. Throughout the paper, the pair $(A,B)$ is assumed to be controllable.
\begin{definition}[Manifest behavior {\cite{Polderman1998}}] \label{def:manifest_b}
    For any $L \in \mathbb{N}$ let 
    \[
    	\mathscr{B}_L \doteq 
    	\left\{
    	\mqty[u_{[0, L - 1]} \\ y_{[0, L - 1]}] \in \mathbb{R}^{L(n_u + n_y)}
        \middle\vert
    	\begin{array}{l}
    		\text{there exist } \, x^0 \in \mathbb{R}^{n_x} \\
            \text{and } x_{[0,L]} \in \mathbb{R}^{(L+1)n_x}  \\
    		\text{such that \eqref{eq:lti_system} holds} \\
    		\text{for all } t \in \mathbb{I}_{[0,L-1]}
    	\end{array} 
    	\right\}
    \]
    be the manifest behavior of \eqref{eq:lti_system} truncated to length $L$. Put differently, $\mathscr{B}_L$ consists of all $L$-long input-output trajectories that can be generated by the LTI system~\eqref{eq:lti_system}. \hfill {\small{$\Box$}} 
\end{definition}
\vspace*{1mm}

The following result is known as the fundamental lemma. It characterizes the subspace of $\mathbb{R}^{L(n_u + n_y)}$ which coincides with $\mathscr{B}_L$. Note that throughout the present paper, observability of \eqref{eq:lti_system} is not assumed since the states are not included in the manifest behavior $\mathscr{B}_L$, cf. Definition~\ref{def:manifest_b}.
\begin{lemma}[{\cite[Theorem~1]{Willems2005}}] \label{lem:fl}
	Consider a pre-recorded input-output trajectory $\vstack(\bar{u}_{[0,T-1]}, \bar{y}_{[0,T-1]}) \in \mathscr{B}_T$, where the input trajectory $\bar{u}_{[0,T-1]}$ is persistently exciting of order $L + n_x$ and the underlying system \eqref{eq:lti_system} is controllable. Then, $\vstack(u_{[0,L-1]}, y_{[0,L-1]}) \in \mathscr{B}_L$ if and only if
	\begin{equation} \label{eq:fl}
		\mqty[u_{[0,L-1]} \\ y_{[0,L-1]}] = \mqty[\mathscr{H}_L(\bar{u}_{[0,T-1]}) \\ \mathscr{H}_L(\bar{y}_{[0,T-1]})] g
	\end{equation}
    holds for some column selection vector $g \in \mathbb{R}^{T-L+1}$.  \hfill {\small{$\Box$}}
\end{lemma}

\subsection{Reproducing Kernel Hilbert Spaces}
Given a non-empty, finite set of $N \in \mathbb{N}$ training data points
\begin{equation} \label{eq:data}
    \mathcal{D} \doteq \left\{(\mathsf{z}_i, \mathsf{y}_i) \in \mathcal{Z} \times \mathcal{Y} \mid i \in \mathbb{I}_{[1,N]} \right\}
\end{equation}
obtained on some $\mathsf{z}$-data domain $\mathcal{Z}$ and $\mathsf{y}$-data co-domain $\mathcal{Y}$, we are interested in learning a function $f: \mathcal{Z} \to \mathcal{Y}$ that fits the data $\mathcal{D}$. To this end, the goodness-of-fit is measured using a loss function $c$ whose value depends on $\mathcal{D}$ and $f$. Additionally, we search only for the functions within some RKHS. For the sake of simplified exposition, we consider the setting with scalar-valued functions, i.e., $\mathcal{Y} \doteq \mathbb{R}$.
\begin{definition}[RKHS \cite{Aronszajn1950, Berlinet2004, Saitoh2016}]
    Let the space $\mathcal{H}$ of functions $\mathcal{Z} \to \mathbb{R}$ be Hilbert space equipped with a real-valued inner product $\expval{\cdot,\cdot}_\mathcal{H}$. If there exists $\kappa: \mathcal{Z} \times \mathcal{Z} \to \mathbb{R}$ such that:
    \begin{itemize}
        \item[i)] For any $\mathsf{z} \in \mathcal{Z}$, the function $\kappa(\cdot, \mathsf{z})$ belongs to $\mathcal{H}$.
        \item[ii)] For any $\mathsf{z} \in \mathcal{Z}$ and any function $g \in \mathcal{H}$, the value $g(\mathsf{z})$ can be reproduced via $g(\mathsf{z}) = \expval{g, \kappa(\cdot, \mathsf{z})}_\mathcal{H}$. \vspace*{2mm}
    \end{itemize}  
    then $\mathcal{H}$ is a reproducing kernel Hilbert space. Item ii) is called the \emph{reproducing property}, whereas the function $\kappa$ itself is the \emph{reproducing kernel}.  \hfill {\small{$\Box$}} 
\end{definition}
\vspace*{1mm}

The correspondence between the RKHS $\mathcal{H}$ and its reproducing kernel $\kappa$ is one-to-one \cite{Aronszajn1950, Saitoh2016}. Furthermore, a function $\kappa: \mathcal{Z} \times \mathcal{Z} \to \mathbb{R}$ is the reproducing kernel of some RKHS $\mathcal{H}$ if and only if for any data set $\mathcal{D} \in (\mathcal{Z} \times \mathbb{R})^N$ as in \eqref{eq:data}, the corresponding $N \times N$ real Gram matrix (Gramian) $\mathbf{K} \doteq \left[\kappa(\mathsf{z}_i, \mathsf{z}_j)\right]_{1 \leq i, j \leq N}$ is symmetric and positive semi-definite. 

Given $\mathcal{D}$ from \eqref{eq:data}, we learn $f$ belonging to the RKHS $\mathcal{H}$ with reproducing kernel $\kappa$ by solving the regression problem
\begin{equation} \label{eq:rkhs_regression}
    \min_{f \in \mathcal{H}} c(\mathcal{D}, \vstack(f(\mathsf{z}_1), \ldots, f(\mathsf{z}_N))) + g\left(\norm{f}_\mathcal{H}\right),
\end{equation}
where $c: (\mathcal{Z} \times \mathbb{R})^N \times \mathbb{R}^N \to \mathbb{R} \cup \{\infty\}$ is the considered loss function, $g: \mathbb{R} \to \mathbb{R}$ is a monotonically non-decreasing regularization, and $\norm{\cdot}_\mathcal{H} \doteq \sqrt{\expval{\cdot,\cdot}_\mathcal{H}}$ is the induced norm on $\mathcal{H}$.
According to the Moore-Aronszajn theorem \cite{Aronszajn1950}, $\mathcal{H}$ contains the dense subspace $\mathrm{span} \lbrace \kappa(\cdot, \mathsf{z}) \rbrace_{\mathsf{z} \in \mathcal{Z}}$ and is thus often infinite-dimensional. However, relying on the seminal representer theorem~\cite{Kimeldorf70a, Schoelkopf01a}, one can guarantee that at least one minimizer $f^\star \in \mathcal{H}$ of \eqref{eq:rkhs_regression} admits a representation as a finite linear combination using the $\mathsf{z}$-data points in $\mathcal{D}$, i.e., 
\begin{equation} \label{eq:rkhs_representation}
    f^\star(\mathsf{z}) \doteq \textstyle\sum\nolimits_{i=1}^N \alpha_i \kappa(\mathsf{z},\mathsf{z}_i).
\end{equation}
The coefficients $\alpha = \vstack(\alpha_1, \ldots, \alpha_N)$ are obtained by substituting \eqref{eq:rkhs_representation} in \eqref{eq:rkhs_regression} and making use of $\norm{f^\star}_\mathcal{H}^2 = \alpha^\top \mathbf{K} \alpha$ \cite{Bishop2006}. This allows to optimize in the finite-dimensional Euclidean space $\mathbb{R}^N$, thereby significantly simplifying problem~\eqref{eq:rkhs_regression}.

\begin{remark} \label{rem:multi-target_regression}
    Statement \eqref{eq:rkhs_representation} can be extended to a vector-valued function $f: \mathcal{Z} \to \mathbb{R}^{n}$ with $n \in \mathbb{N}$, i.e., $f \doteq \vstack(f_1, \ldots, f_n)$, where $f_j : \mathcal{Z} \to \mathbb{R}$, $\forall \, j \in \mathbb{I}_{[1,n]}$. Suppose the data $\mathcal{D}$ in \eqref{eq:data} is given with the co-domain $\mathcal{Y} \doteq \mathbb{R}^n$. Further, assume that each component $f_j$ is a member of the RKHS $\mathcal{H}$, whose kernel $\kappa$ is the same as in \eqref{eq:rkhs_regression} and \eqref{eq:rkhs_representation}. Next, construct $n$ separate problems of the form \eqref{eq:rkhs_regression}, each having $f_j, j \in \mathbb{I}_{[1,n]}$ as its decision variable. Relying on the representer theorem, we express the corresponding minimizers via $f_j^\star(\mathsf{z}) \doteq \sum_{i=1}^N \alpha_{ij} \kappa(\mathsf{z},\mathsf{z}_i)$, where $\alpha \doteq [\alpha_1 \, \cdots \, \alpha_n]$ is the $N \times n$ coefficient matrix filled with vectors $\alpha_j \doteq \vstack(\alpha_{1j}, \ldots, \alpha_{Nj}) \in \mathbb{R}^N$. \hfill {\small{$\Box$}}
\end{remark}

In practice, for $i \neq j$, the components $f_i$ and $f_j$ might lie in different RKHS spanned by $\kappa_i$ and $\kappa_j$, respectively. In this case, one could use the compound kernel $\kappa(\cdot,\cdot) \doteq \sum_{i=1}^n \kappa_i(\cdot,\cdot)$.

\subsection{Kernel-based Regression} 
The key idea behind LTI embeddings proposed in \cite{Markovsky2023} is to represent the system dynamics as a linear combination of nonlinear basis functions (or \textit{features}). Inspired by this concept, we continue by recapping feature-based regression. Later, we draw the connection to kernel-based regression in RKHS. 

Since one is often concerned with predicting multiple future outputs and/or multiple entries of the same output, we consider the vector-valued setting with $f: \mathcal{Z} \to \mathcal{Y} \doteq \mathbb{R}^{n}$, cf. Remark~\ref{rem:multi-target_regression}.

Suppose that $f$ admits a linear representation $f(\mathsf{z}) = \theta^\top \phi(\mathsf{z}), \forall \, \mathsf{z} \in \mathcal{Z}$ with the parameter matrix $\theta \in \mathbb{R}^{n_\theta \times n}$ and the feature map $\phi : \mathcal{Z} \to \mathbb{R}^{n_\theta}$. Then, in the least-squares sense, the optimal parameter $\theta^\star$ specifying  $f^\star(\mathsf{z}) = \theta^{\star \top} \phi(\mathsf{z})$, is found as
\begin{equation} \label{eq:regression}
	\theta^\star \in \arg\min_{\theta \in \mathbb{R}^{n_\theta \times n}} \; \norm{\Phi \theta - \mathbf{y}}_F^2,
\end{equation} 
where the $i$-th row of $\Phi \in \mathbb{R}^{N \times n_\theta}$ is defined as $\phi(\mathsf{z}_i)^\top$, while the $i$-th row of $\mathbf{y} \in \mathbb{R}^{N \times n}$ is given by $\mathsf{y}_i^\top$.
With the help of the Moore-Penrose pseudo-inverse operator $(\cdot)^\dagger$, a minimizer $\theta^\star$ can be calculated in closed form via $\theta^\star \doteq \Phi^\dagger \mathbf{y}$.

Alternatively, using the \emph{kernel trick}, i.e., the ansatz $\theta \doteq \Phi^\top \alpha$ with $\alpha \in \mathbb{R}^{N \times n}$ \cite[Chapter~6]{Bishop2006}, \eqref{eq:regression} can be rewritten as
\begin{equation} \label{eq:regression_kernel_trick}
	\alpha^\star \in \arg\min_{\alpha \in \mathbb{R}^{N \times n}} \; \norm{\mathbf{K} \alpha - \mathbf{y}}_F^2.
\end{equation}
Here, one of the minimizers can be stated as $\alpha^\star \doteq \mathbf{K}^\dagger \mathbf{y}$ in terms of the (possibly singular) Gram matrix $\mathbf{K} \doteq \Phi \Phi^\top$ of the  reproducing kernel $\kappa(\mathsf{z}, \mathsf{z}^\prime) \doteq \phi(\mathsf{z})^\top \phi(\mathsf{z}^\prime)$. Consequently, the linear combination $f(\mathsf{z}) = \theta^\top \phi(\mathsf{z})$ can be written as
\begin{equation} \label{eq:reg_kernel_sol}
    f^\star(\mathsf{z}) \doteq \mathbf{y}^\top \mathbf{K}^\dagger \mathbf{k}(\mathsf{z}),
\end{equation} 
where $\mathbf{k}(\mathsf{z}) \doteq \mqty[\kappa(\mathsf{z}, \mathsf{z}_1) & \cdots & \kappa(\mathsf{z}, \mathsf{z}_N)]^\top$ is the kernel $\kappa$, centered at all $\mathsf{z}$-data points in $\mathcal{D}$, cf. \eqref{eq:rkhs_representation} and Remark~\ref{rem:multi-target_regression}. The recent work \cite{Lazar2023} views $\mathbf{k}(\mathsf{z})$ as a data-dependent feature vector.

Following the original fundamental lemma \cite{Willems2005}, we stick to the noise-free setting. Put differently, given the true function $f$, we assume that the data $\mathcal{D}$ is exact, i.e., $f(\mathsf{z}_i) = \mathsf{y}_i, \forall \, i \in \mathbb{I}_{[1,N]}$. Similarly to Remark~\ref{rem:multi-target_regression}, we also assume that each $f_j, j \in \mathbb{I}_{[1,n]}$ of $f: \mathcal{Z} \to \mathbb{R}^{n}$ belongs to RKHS $\mathcal{H}$ with the kernel $\kappa$.

Since $\mathcal{D}$ is assumed to contain noise-free data, we do not consider regularization in \eqref{eq:regression} and \eqref{eq:regression_kernel_trick}. Instead, the approach in \eqref{eq:reg_kernel_sol} relies on pseudo-inverse regularization, which is often applied for GPs~\cite{Mohammadi2017}. Based on the data $\mathcal{D}$, for each $z \in \mathcal{Z}$, the GP regression with pseudo-inverse regularization delivers a mean estimate of $f(\mathsf{z})$, which coincides with $f^\star(\mathsf{z})$ from \eqref{eq:reg_kernel_sol}. 

An additional useful feature of the GP perspective is the posterior variance estimate
\begin{equation} \label{eq:gp_variance}
    \sigma^2(\mathsf{z}) \doteq \kappa(\mathsf{z}, \mathsf{z}) - \mathbf{k}(\mathsf{z})^\top \mathbf{K}^\dagger \mathbf{k}(\mathsf{z}).
\end{equation}
This variance describes the uncertainty surrounding each entry $f_j^\star(\mathsf{z})$ of $f^\star(\mathsf{z})$ due to lack of exploration of the RKHS, see Appendix\ref{app:gp_proof}. Observe that $\sigma^2(\mathsf{z})$ depends only on the choice of $\kappa$ and on the $\mathsf{z}$-data in $\mathcal{D}$. Hence, $\sigma^2(\mathsf{z})$ stays the same for each index $j \in \mathbb{I}_{[1,n]}$. Next we extend a result from \cite{Fasshauer2011, Hu2023} to singular Gramians; its proof is given in Appendix\ref{app:error_proof}.
\begin{lemma}[RKHS error bound with singular Gramian] \label{lem:absolute_error}
    Let $f: \mathcal{Z} \to \mathcal{Y} \doteq \mathbb{R}^{n}$ be unknown with $f \doteq \vstack(f_1, \ldots, f_n)$. Assume that $f_j \in \mathcal{H}, \forall \, j \in \mathbb{I}_{[1,n]}$, where $\mathcal{H}$ is a RKHS and $\kappa$ is the corresponding reproducing kernel, cf. Remark~\ref{rem:multi-target_regression}. Then, for the data $\mathcal{D}$ from \eqref{eq:data}, the estimate $f^\star$ from \eqref{eq:reg_kernel_sol} satisfies
    \begin{equation} \label{eq:abs_error_bound}
        \norm{f(\mathsf{z}) - f^\star(\mathsf{z})}_2^2 \leq \sigma^2(\mathsf{z}) \textstyle\sum\nolimits_{j=1}^n \norm{f_j}^2_\mathcal{H}, \quad \forall \, \mathsf{z} \in \mathcal{Z}. \vspace*{-2mm}
    \end{equation}
    \hfill {\small{$\Box$}}
\end{lemma}
Note that $\sigma^2(\mathsf{z})$ in \eqref{eq:abs_error_bound} can be factored out since each $f_j$ is in $\mathcal{H}$.

\subsection{Kernel Representation of the Fundamental Lemma}
We consider a time-invariant linear or nonlinear system described by a difference equation with finite lag $\ell \in \mathbb{N}_0$ \cite{Markovsky2023}. Following \cite{Huang2023}, we assume temporarily that an $L$-long output trajectory $y_{[\ell,\ell+L-1]}$ of the system can be uniquely deduced from: i) a corresponding input sequence $u_{[\ell,\ell+L-1]}$; ii) an initializing input-output trajectory $(u_{[0, \ell - 1]}, y_{[0, \ell - 1]})$. Thus, we define the concatenation $\mathsf{z} \doteq \vstack(u_{[0, \ell+L-1]}, y_{[0, \ell-1]})$ and set $\mathcal{Z} \doteq \mathbb{R}^{(\ell + L)n_u + \ell n_y}$ accordingly. Suppose that each entry of $y_{[\ell,\ell+L-1]}$ in $\mathcal{Y} \doteq \mathbb{R}^{L n_y}$ can be expressed as a linear combination of features in $\phi: \mathcal{Z} \to \mathbb{R}$. Then, using $\kappa(\cdot,\cdot) \doteq \phi(\cdot)^\top \phi(\cdot)$, we arrive at the explicit kernel predictor
\begin{equation} \label{eq:nl_dd_predictor_explicit}
	y_{[\ell,\ell+L-1]} = \mathscr{H}_L(\bar{y}_{[\ell, T-1]}) \mathbf{K}^\dagger \mathbf{k}(\mathsf{z}).
\end{equation}
Here, $(\bar{u}_{[0,T-1]}, \bar{y}_{[0, T-1]})$ is the $T$-long pre-recorded input-output trajectory, cf. \eqref{eq:reg_kernel_sol}. The training data needed for the kernel evaluations has a special pattern: It comprises the $N \doteq T - L - \ell + 1$ columns of the stacked Hankel matrix
\[
	\mqty[\mathsf{z}_1 & \cdots & \mathsf{z}_N] \doteq \mqty[\mathscr{H}_{\ell + L}(\bar{u}_{[0, T-1]}) \\ \mathscr{H}_\ell(\bar{y}_{[0, T-L-1]})].
\]
With $g \in \mathbb{R}^{N}$, the implicit counterpart to \eqref{eq:nl_dd_predictor_explicit} reads
\begin{equation} \label{eq:nl_dd_predictor}
	\mqty[\mathbf{k}(\mathsf{z}) \\ y_{[\ell,\ell+L-1]}] = \mqty[\mathbf{K} \\ \mathscr{H}_L(\bar{y}_{[\ell, T-1]})]g.	
\end{equation}
As noted in \cite{Huang2023}, even though the above equation closely resembles the fundamental lemma~\eqref{eq:fl}, there is no guarantee that for a specified $g \in \mathbb{R}^{N}$ the left hand side variables $\mathsf{z} \doteq \vstack(u_{[0, \ell+L-1]}, y_{[0, \ell-1]})$ and $y_{[\ell,\ell+L-1]}$ constitute a valid $(\ell + L)$-long input-output trajectory of the system.
\begin{remark}[Equivalence of kernel predictors]
    The formulations \eqref{eq:nl_dd_predictor_explicit} and \eqref{eq:nl_dd_predictor} are equivalent if and only if the Gram matrix $\mathbf{K}$ is of full rank. If this is not the case, the formulation \eqref{eq:nl_dd_predictor} allows for $y_{[\ell,\ell+L-1]}$ to lie in the affine subspace
    \[
        \left\lbrace\mathscr{H}_L(\bar{y}_{[\ell, T-1]}) \left[\mathbf{K}^\dagger \mathbf{k}(\mathsf{z}) + (I - \mathbf{K}^\dagger \mathbf{K})v\right] \mid \forall \, v \in \mathbb{R}^N \right\rbrace.
    \]
    However, \eqref{eq:nl_dd_predictor_explicit} only delivers a unique solution for $v \doteq 0$ in the above expression. We also note that the $\mathsf{z}$-data points cannot all be chosen freely as the initializing output trajectory $y_{[0,\ell-1]}$ (in most cases) depends on the input sequence $u_{[0,\ell-1]}$. This may lead to rank deficiency in $\mathbf{K}$; in \cite{Huang2023} this is addressed via ridge regression (i.e., quadratic regularization). \hfill {\small{$\Box$}}
\end{remark}
\vspace*{1mm}

Henceforth, we work with the structure \eqref{eq:nl_dd_predictor} proposed by \cite{Huang2023} since it is linear in $g$ as opposed to \cite[Theorem~4]{Lian2021}, which derives a quadratic form in $g$. In contrast to \cite{Huang2023}, we aim for results ensuring the exact representation of system trajectories.

\section{Main Results} \label{sec:results}
We begin by translating the classic fundamental lemma (Lemma~\ref{lem:fl}) with the same requirements on persistency of excitation of the measured input trajectory to a kernel setting. 

Consider an input-output trajectory $\vstack(\bar{u}_{[0,T-1]}, \bar{y}_{[0,T-1]}) \in \mathscr{B}_T$ of \eqref{eq:lti_system}. For any $L \in \mathbb{N}$, let $\kappa_L: (\mathsf{z}, \mathsf{z}^\prime) \in \mathcal{Z} \times \mathcal{Z} \mapsto \mathsf{z}^\top \mathsf{z}^\prime \in \mathbb{R}$ be the linear kernel with $\mathcal{Z} \doteq \mathbb{R}^{L n_u}$. Furthermore, let the $\mathsf{z}$-data of $\mathcal{D}$ comprise the columns $\mathsf{z}_{i+1} \doteq \bar{u}_{[i,L+i-1]}$, $\forall \, i \in \mathbb{I}_{[0,T-L]}$ of $\mathscr{H}_L(\bar{u}_{[0, T-1]})$. For $\kappa_L$, we define both the corresponding Gram matrix $\mathbf{K}_L$ and the vector $\mathbf{k}(\mathsf{z}) \doteq \vstack(\kappa_L(\mathsf{z}, \mathsf{z}_1), \ldots, \kappa_L(\mathsf{z}, \mathsf{z}_{T-L+1}))$.
\begin{lemma}[Kernelized fundamental lemma – LTI case] \label{lem:fl_kernelized}	~\\
	Consider  $\kappa_L(\mathsf{z}, \mathsf{z}^\prime) = \mathsf{z}^\top \mathsf{z}^\prime$ and let $\rank\mathbf{K}_{L+n_x} = (L+n_x)n_u$. Suppose that \eqref{eq:lti_system} is controllable. Then the following holds:
    \begin{itemize}
        \item[i)] $\vstack(u_{[0, L-1]}, y_{[0, L-1]}) \in \mathscr{B}_L$ if and only if there exists a vector $g \in \mathbb{R}^{T-L+1}$ such that
	    \begin{equation} \label{eq:fl_kernelized}
		      \mqty[\mathbf{k}(u_{[0, L-1]}) \\ y_{[0, L-1]}] = \mqty[\mathbf{K}_L \\ \mathscr{H}_L(\bar{y}_{[0, T-1]})] g.
	    \end{equation}
        \item[ii)] Any tuple $(u_{[0, L-1]}, y_{[0, L-1]}, g) \in \mathbb{R}^{L n_u} \times \mathbb{R}^{L n_y} \times \mathbb{R}^{T-L+1}$ that satisfies \eqref{eq:fl} solves \eqref{eq:fl_kernelized} and vice versa. \hfill {\small{$\Box$}}
    \end{itemize}
\end{lemma}
    
\begin{proof}
	Since $\kappa_L(\mathsf{z}, \mathsf{z}^\prime) = \mathsf{z}^\top \mathsf{z}^\prime$, the Gram matrix can be expressed as $\mathbf{K}_{L+n_x} = \mathscr{H}_{L+n_x}(\bar{u}_{[0, T-1]})^\top \mathscr{H}_{L+n_x}(\bar{u}_{[0, T-1]})$. Furthermore, since $\rank A = \rank A^\top A$ for any matrix $A$, $\bar{u}_{[0, T-1]}$ is persistently exciting of order $L + n_x$. Then, according to Lemma~\ref{lem:fl}, for any $\vstack(u_{[0, L-1]}, y_{[0, L-1]}) \in \mathscr{B}_L$ there exists $g \in \mathbb{R}^{T-L+1}$ that satisfies \eqref{eq:fl}. To prove Assertion i), we multiply $\mathscr{H}_{L}(\bar{u}_{[0, T-1]})^\top$ to the first $L n_u$ rows of \eqref{eq:fl} both to the left and right hand side, obtaining
	\[
	       \mqty[\mathscr{H}_{L}(\bar{u}_{_{[0, T-1]}})^\top u_{_{[0, L-1]}} \\ y_{_{[0, L-1]}}] = \mqty[\mathscr{H}_{L}(\bar{u}_{_{[0, T-1]}})^\top \mathscr{H}_L(\bar{u}_{_{[0, T-1]}}) \\ \mathscr{H}_L(\bar{y}_{_{[0, T-1]}})] g.
    \]  
	Due to the chosen $\kappa_L$, the above equation is equivalent to \eqref{eq:fl_kernelized}. Since $\bar{u}_{[0, T-1]}$ is persistently exciting of order $L+n_x$, $\rank \mathscr{H}_{L}(\bar{u}_{[0, T-1]}) = L n_u$; hence,  any $h \in \mathbb{R}^{L n_u}$ can be uniquely deduced from $\mathscr{H}_{L}(\bar{u}_{[0, T-1]})^\top h$, cf. Assertion ii).
\end{proof}
\vspace*{1mm}

We note that the system \eqref{eq:fl_kernelized} contains more equations than \eqref{eq:fl} since---due to the persistency of excitation requirements---$\mathscr{H}_{L}(\bar{u}_{[0, T-1]})$ has more columns than rows. Thus,  \eqref{eq:fl_kernelized} is not directly computationally preferable to \eqref{eq:fl} in Lemma~\ref{lem:fl}.

\subsection{Exact Kernel Predictors for Nonlinear Systems}
Lemma~\ref{lem:fl_kernelized} can be used to derive alternative kernelized representations for some of the fundamental lemma extensions mentioned in the Introduction. For instance, the data-driven description of Hammerstein systems can be addressed via
\begin{equation} \label{eq:kernel}
    \kappa_L(u_{[0,L-1]}, u_{[0,L-1]}^\prime) \doteq \textstyle\sum\nolimits_{i=0}^{L-1} \kappa_u(u(i), u^\prime(i)),
\end{equation}
where $\kappa_u : \mathbb{R}^{n_u} \times \mathbb{R}^{n_u} \to \mathbb{R}$ is a reproducing kernel, whose argument is a pair of inputs $u(i), u^\prime(i) \in \mathbb{R}^{n_u}$~\cite{Berberich2020, Lian2021}.

Consider a discrete-time Hammerstein system, where the input $u(t) \in \mathbb{R}^{n_u}$ is fed through a nonlinear function $f: \mathbb{R}^{n_u} \to \mathbb{R}^{n_f}$, and, subsequently, through the LTI system from \eqref{eq:lti_system} with $B \in \mathbb{R}^{n_x \times n_f}$ and $D \in \mathbb{R}^{n_y \times n_f}$:
\begin{equation}
	\label{eq:hammerstein_system}
	\left\{
	\begin{array}{l l}
		x(t+1) &= A x(t) + B f(u(t)), \quad x(0) \doteq x^0,\\
		y(t) &= C x(t) + D f(u(t)). 
	\end{array}
	\right.
\end{equation}
Suppose we collect a $T$-long input-output trajectory $(\bar{u}_{[0,T-1]}$$,$ $\bar{y}_{[0,T-1]})$ generated by \eqref{eq:hammerstein_system}. Next, we provide a kernelized version of \cite[Proposition~5]{Berberich2020}.

\begin{lemma}[Hammerstein systems] \label{lem:fl_k_hammerstein}
    Assume there exist both a $n_\theta \times n_f$ real-valued matrix $\theta$ and a (possibly nonlinear) map $\phi_u: \mathbb{R}^{n_u} \to \mathbb{R}^{n_\theta}$ such that $f(u) = \theta^\top \phi_u(u), \forall \, u \in \mathbb{R}^{n_u}$ and let the corresponding kernel $\kappa_L$ be constructed via \eqref{eq:kernel} using $\kappa_u : (u, u^\prime) \in \mathbb{R}^{n_u} \times \mathbb{R}^{n_u} \mapsto \phi_u(u)^\top \phi_u(u^\prime) \in \mathbb{R}$. Suppose that $(A,B)$ is controllable and that  $\rank \theta = n_f$.

	If, for the collected data, $\rank \mathbf{K}_{L+n_x} = (L+n_x)n_\theta$, then the pair $(u_{[0,L-1]}, y_{[0,L-1]})$ is an $L$-long input-output trajectory of \eqref{eq:hammerstein_system} if and only if there exists $g \in \mathbb{R}^{T-L+1}$ that satisfies \eqref{eq:fl_kernelized} with the kernel from \eqref{eq:kernel}. \hfill {\small{$\Box$}}
\end{lemma}
	
\begin{proof}
	The sum \eqref{eq:kernel} can be expressed in terms of the inner product $\phi(u_{[0,L-1]})^\top \phi(u_{[0,L-1]}^\prime)$, where $\phi(u_{[0,L-1]}) = \vstack(\phi_u(u(0)), \ldots, \phi_u(u(L-1)))$ is a concatenation of individual inputs wrapped in $\phi_u$. Let the auxiliary signal $v: \mathbb{Z} \to \mathbb{R}^{n_\theta}$ be defined as $v(t) \doteq \phi_u(u(t)) \in \mathbb{R}^{n_\theta}$, $\forall \, t \in \mathbb{N}_0$. Then, $\kappa_L(u_{[0,L-1]}, u_{[0,L-1]}^\prime) = v_{[0,L-1]}^\top v_{[0,L-1]}^\prime$, where $v_{[0, L - 1]} \doteq \left[\phi_u(u(i))\right]_{i=0}^{L-1}$ and $v^\prime_{[0, L - 1]} \doteq \left[\phi_u(u^\prime(i))\right]_{i=0}^{L-1}$, i.e., $\kappa_L$ turns into the linear kernel, identical to the one in Lemma~\ref{lem:fl_kernelized}. From $\rank \mathbf{K}_{L+n_x} = (L+n_x)n_\theta$ we conclude that $\bar{v}_{[0,T-1]} \doteq \left[\phi_u(\bar{u}(i))\right]_{i=0}^{T-1}$ is persistently exciting of order $L+n_x$. The proof then follows from Lemma~\ref{lem:fl_kernelized}, wherein we substitute the input trajectories $u_{[0,L-1]}$ and $\bar{u}_{[0,T-1]}$ with the auxiliary signal trajectories $v_{[0,L-1]}$ and $\bar{v}_{[0,T-1]}$, respectively.
\end{proof}
\vspace*{2mm}

\begin{table*} 
	\caption{Exact kernelized representations of the fundamental lemma for different linear and nonlinear system classes.}
	\begin{center}
		\label{tab:fl_lemmas}
		\begin{tabular}{p{0.17\linewidth} p{0.35\linewidth} p{0.4\linewidth}}  
			\hline
			&&\\[-0.8em]
			System class & Kernelized fundamental lemma & Underlying kernel structure \\
			\hline
			&&\\[-0.8em]
			Discrete-time LTI & $\mqty[\mathbf{k}(u_{[0, L-1]}) \\ y_{[0, L-1]}] = \mqty[\mathbf{K}_L \\ \mathscr{H}_L(\bar{y}_{[0, T-1]})] g$ & $\kappa_L(u_{[0,L-1]}, u_{[0,L-1]}^\prime) \doteq u_{[0,L-1]}^\top u_{[0,L-1]}^\prime$ \\
			&&\\[-0.8em]
			\multirow[t]{2}{*}{Hammerstein} & $\mqty[\mathbf{k}(u_{[0, L-1]}) \\ y_{[0, L-1]}] = \mqty[\mathbf{K}_L \\ \mathscr{H}_L(\bar{y}_{[0, T-1]})] g$ & $\kappa_L(u_{[0,L-1]}, u_{[0,L-1]}^\prime) \doteq \sum_{i=0}^{L-1} \phi_u(u(i))^\top \phi_u(u^\prime(i))$ \\
			&&\\[-0.8em]
			Differentially flat (SISO) & $\mqty[\mathbf{k}(u_{[0,L-n_x-1]}, y_{[0,L-2]}) \\ y_{[0, L-1]}] = \mqty[\mathbf{K}_{L - n_x} \\ \mathscr{H}_L(\bar{y}_{[0, T-1]})] g$ & 
			$\begin{gathered}
				\kappa_{L - n_x}(u_{_{[0,L - n_x - 1]}}, y_{_{[0, L - 2]}}, u_{_{[0,L - n_x - 1]}}^\prime, y_{_{[0, L - 2]}}^\prime) \doteq \\ \textstyle\sum_{i=0}^{L - n_x - 1} \phi(u(i), y_{[i,i+n_x-1]})^\top \phi(u(i)^\prime, y_{[i,i+n_x-1]}^\prime)
			\end{gathered}$ \\
			&&\\[-0.8em]
			\hline
		\end{tabular}
	\end{center}
\end{table*}

\begin{lemma}[Posterior variance] \label{lem:zero_variance}
    Let $(\bar{u}_{[0,T-1]}, \bar{y}_{[0,T-1]})$ be a pre-recorded trajectory of \eqref{eq:hammerstein_system} with $\bar{v}_{[0,T-1]}$ persistently exciting of order $L$ and consider the setting from Lemma~\ref{lem:fl_k_hammerstein}. Then, for every $L$-long input trajectory $u_{[0, L - 1]}$ of \eqref{eq:hammerstein_system}, we have $\sigma^2(u_{[0, L - 1]}) = 0$ for the posterior variance \eqref{eq:gp_variance}. 
    \hfill {\small{$\Box$}}
\end{lemma}

\begin{proof}
    Let $u_{[0, L - 1]}$ be any $L$-long input trajectory of \eqref{eq:hammerstein_system}. We construct $v_{[0, L - 1]} \doteq \left[\phi_u(u(i))\right]_{i=0}^{L-1}$. Since $\mathscr{H}_L(\bar{v}_{[0,T-1]})$ is of full row rank, there exist $g_1 \in \mathbb{R}^{T-L+1}$ and $g_2 \in \mathbb{R}$ solving
    \[
        v_{[0, L - 1]} = \mqty[\mathscr{H}_L(\bar{v}_{[0,T-1]}) & v_{[0, L - 1]}] \vstack(g_1, g_2).
    \]
    Multiplying both sides of the above equation from the left with $\mqty[\mathscr{H}_L(\bar{v}_{[0,T-1]}) & v_{[0, L - 1]}]^\top$ we obtain
    \[
        \mqty[\mathbf{k}(u_L) \\ \kappa_L(u_L, u_L)] = \mqty[\mathbf{K}_L & \mathbf{k}(u_L) \\ \mathbf{k}(u_L)^\top & \kappa_L(u_L, u_L)] \mqty[g_1 \\ g_2],
    \]
    wherein we use the shorthand variable $u_L \doteq u_{[0, L-1]}$. Since $\rank \mathscr{H}_L(\bar{v}_{[0,T-1]}) = L n_\theta$, the column span of $\mathscr{H}_L(\bar{v}_{[0,T-1]})$ contains $v_{[0, L - 1]}$. Therefore, $\mathbf{k}(u_L) = \mathscr{H}_L(\bar{v}_{[0,T-1]})^\top v_{[0, L - 1]}$ is in the column span of $\mathbf{K}_L = \mathscr{H}_L(\bar{v}_{[0,T-1]})^\top \mathscr{H}_L(\bar{v}_{[0,T-1]})$ meaning that the linear system $\mathbf{K}_L g_1 = \mathbf{k}(u_L)$ has at least one solution. Henceforth, instead of $g_1 = 0$ and $g_2 = 1$, we pick $g_1 \doteq (1/2) \cdot \mathbf{K}_L^\dagger \mathbf{k}(u_L)$ and $g_2 \doteq 1/2$. Then, expanding the second block row of the above equation, we have
    \[
        \kappa_L(u_L, u_L) = (\nicefrac{1}{2}) \cdot \mathbf{k}(u_L)^\top \mathbf{K}_L^\dagger \mathbf{k}(u_L) + (\nicefrac{1}{2}) \cdot \kappa_L(u_L, u_L).
    \]
    Hence, $\kappa_L(u_L, u_L) - \mathbf{k}(u_L)^\top \mathbf{K}_L^\dagger \mathbf{k}(u_L) = 0$, cf. \eqref{eq:gp_variance}. 
\end{proof}
\vspace*{1mm}

Similar to \cite{Alsalti2021}, we consider SISO flat nonlinear systems
\begin{equation} \label{eq:siso_flat}
    \left\{
	\begin{array}{l l}
		x(t+1) &= f(x(t), u(t)), \quad x(0) \doteq x^0,\\
		y(t) &= h(x(t)), 
	\end{array}
	\right.
\end{equation}
with $x(t) \in \mathbb{R}^{n_x}$, $u(t) \in \mathbb{R}$, and $y(t) \in \mathbb{R}$, whereas $f: \mathbb{R}^{n_x} \times \mathbb R \to \mathbb R^{n_x}$ and $h: \mathbb R^{n_x} \to \mathbb R$ are smooth functions with $f(0,0) = 0$ and $h(0) = 0$. As shown in \cite{Alsalti2021}, under suitable assumptions, the above model can be cast into LTI dynamics
\begin{equation}
    \left\{
	\begin{array}{l l}
		\xi(t+1) &= A_c\xi(t) + B_c v(t),\\
		y(t) &= C_c \xi(t), 
	\end{array}
	\right.
\end{equation}
where $\xi(t) \doteq y_{[t, t+n_x-1]} \in \mathbb R^{n_x}$ is the state $x(t)$ after transformation, $v(t) \in \mathbb R$ is the synthetic input which results from the inverse of the static state feedback control law, and the matrix triple $(A_c, B_c, C_c)$ appears in Brunovsky canonical form. Following \cite{Alsalti2021}, we assume that $v(t)$ admits a linear combination $v(t) \doteq \theta^\top \phi(u(t), y_{[t,t+n_x-1]})$, $\forall \, t \in \mathbb{N}_0$ with $n_\theta$ terms in $\theta$ and $n_\theta$ features in $\phi$. Next we define the kernel $\kappa_{L-n_x}$ for some $L > n_x$ as
\begin{equation} \label{eq:kernel_flat}
    \begin{split}
        &\kappa_{L - n_x}(u_{[0,L - n_x - 1]}, y_{[0, L - 2]}, u_{[0,L - n_x - 1]}^\prime, y_{[0, L - 2]}^\prime) \doteq \\
        &\textstyle\sum\nolimits_{i=0}^{L - n_x - 1} \kappa_u(u(i), y_{[i, i + n_x - 1]}, u^\prime(i), y_{[i, i + n_x - 1]}^\prime),
    \end{split}
\end{equation}
where $\kappa_u(\cdot, \cdot, \cdot, \cdot) \doteq \phi(u(i), y_{[i,i+n_x-1]})^\top \phi(u^\prime(i), y_{[i,i+n_x-1]}^\prime)$. Again, assuming that $\rank \theta = 1$, we have a controllable pair $(A_c, B_c \theta^\top)$. At this point, \cite{Alsalti2021} leverages the original fundamental lemma \eqref{eq:fl} to derive its extension for SISO flat nonlinear systems \eqref{eq:siso_flat}.

Suppose that $(\bar{u}_{[0, T - n_x - 1]}, \bar{y}_{[0, T-1]})$ is a pre-recorded input-output trajectory of \eqref{eq:siso_flat}. Let $\mathbf{K}_L$ be the Gram matrix of $\kappa_L$ from \eqref{eq:kernel_flat} evaluated at $\mathsf{z}_{i+1} \doteq (\bar{u}_{[i, i+L-1]}, \bar{y}_{[i, i+L+n_x-2]})$, $i \in \mathbb{I}_{[0,T - L - n_x]}$. Similarly, let $\mathbf{K}_{L - n_x}$ be the Gram matrix of $\kappa_{L-n_x}$ from \eqref{eq:kernel_flat} evaluated at $\mathsf{z}_{i+1} \doteq (\bar{u}_{[i, i+L-n_x-1]},$ $\bar{y}_{[i, i+L-2]})$, $i \in \mathbb{I}_{[0,T - L]}$.
\begin{lemma}[Kernel predictor for SISO flat systems] \label{lem:fl_k_flat}
    If with the kernel \eqref{eq:kernel_flat}, the Gram matrix satisfies $\rank \mathbf{K}_L = L n_\theta$ for $(\bar{u}_{[0, T - n_x - 1]}$, $\bar{y}_{[0, T-1]})$, then $(u_{[0,L-n_x-1]}, y_{[0,L-1]})$ is a trajectory of the system \eqref{eq:siso_flat} if and only if there exists a vector $g \in \mathbb{R}^{T-L+1}$ such that
    \[
        \mqty[\mathbf{k}(u_{[0,L-n_x-1]}, y_{[0,L-2]}) \\ y_{[0, L-1]}] = \mqty[\mathbf{K}_{L - n_x} \\ \mathscr{H}_L(\bar{y}_{[0, T-1]})] g,
    \]
    where $\mathbf{k}(u_{[0,L-n_x-1]}, y_{[0,L-2]})$ stacks $\kappa_{L-n_x}(\bar{u}_{[i, i+L-n_x-1]},$ $\bar{y}_{[i, i+L-2]}, u_{[0,L-n_x-1]}, y_{[0,L-2]})$ for all $i \in \mathbb{I}_{[0,T - L]}$. \hfill {\small{$\Box$}}
\end{lemma}   

\begin{proof}
    Similar to the one of Lemma~\ref{lem:fl_k_hammerstein}; we utilize Equation 12 from \cite[Proposition~1]{Alsalti2021} and then perform a linear transformation as used in the proof of Lemma~\ref{lem:fl_kernelized}.
\end{proof}
\vspace*{1mm}

\begin{remark}[Kernel rank for infinite feature dimension]\label{rem:kernel_inf}
	For most of the popular kernel functions that can be chosen in place of $\kappa_u(\cdot,\cdot)$, the underlying feature space is infinite-dimensional and so is the corresponding RKHS. In the limit for $n_\theta \to \infty$, the condition $\rank \mathbf{K}_{L+n_x} = (L+n_x)n_\theta$ will never be satisfied using a finite amount of data $N = \left\vert \mathcal{D} \right\vert$. Still, the output prediction performance can only improve as the rank of $\mathbf{K}_{L+n_x} \in \mathbb{R}^{N \times N}$ is non-decreasing with $N$ \cite{Lian2021}. \hfill {\small{$\Box$}}
\end{remark}

\begin{remark}[Link to the representer theorem]
    At this point, it is helpful to put the preceding results into perspective. Recall that the seminal representer theorem gives that the minimizer of a regression in a RKHS leads to a linear combination of kernel evaluations~\eqref{eq:rkhs_representation}. Using the so-called kernel trick and the feature matrices $\Phi$, one obtains the usual explicit kernel predictor \eqref{eq:reg_kernel_sol}, which also resembles the structure of mean GP regression. The results of Lemma \ref{lem:fl_k_hammerstein} and   \ref{lem:fl_k_flat} are cases for which the linear ansatz of the representer theorem leads to exact models of nonlinear systems which (at least for noise-free data)  generalize perfectly to new data. Moreover, our results underpin that the \textit{natural} nonlinear extension of persistency of excitation are rank conditions on the kernel-based Gramians; see also \cite{Lazar2023} for a similar insight in feature spaces. \hfill {\small{$\Box$}}
\end{remark}

The (nonlinear) system classes discussed in this paper as well as their kernelized fundamental lemma representations are summarized in Table~\ref{tab:fl_lemmas}. In the rightmost column, we list the kernels that yield exact representations provided that the feature space is finite-dimensional, cf. Remark~\ref{rem:kernel_inf}. For the considered cases, the second equation of the kernelized implicit predictors is identical to the original fundamental lemma~\eqref{eq:fl}. This reflects the linearity of~\eqref{eq:regression_kernel_trick} in $\mathbf{y}$ induced by the kernel trick.
\begin{remark}[Kernel and feature selection]   
    Although finite-dimensional RKHS expansions are in general not known, Lemma~\ref{lem:zero_variance} opens an avenue towards kernel selection by minimizing the posterior variance \eqref{eq:gp_variance} for some set of interest $\mathcal{Z}_0 \subseteq \mathcal{Z}$ and over a class of kernel functions. Both the implementation and efficacy of this approach remain to be clarified in the future. A major difference of our kernel regression perspective to  \cite{Lazar2023}, is that our results suggest searching over a single kernel function $\kappa$ as opposed to selecting features comprising $\phi$ \cite{Lazar2023}. Furthermore, note that the posterior variance \eqref{eq:gp_variance} can be made arbitrarily small by generating a dense set of data points (i.e, trajectories). The condition $\sigma(\mathsf{z}) = 0$, however, strongly depends on the considered class of kernel functions.
    A comparison of both approaches is subject of future work. \hfill {\small{$\Box$}}
\end{remark}

\section{Conclusion} \label{sec:conclusion}
This paper discussed the links between the fundamental lemma and kernel regression. For the original fundamental lemma by Willems and co-authors, an equivalent, kernelized result can be obtained using a trivial, linear kernel. We also show that some of the existing fundamental lemma generalizations for nonlinear systems admit respective kernelized versions. This leads to the observation that some nonlinear extensions of the fundamental lemma can be understood as special (exact) cases of the representer theorem for regression problems in reproducing kernel Hilbert spaces. Future work will discuss the extension to noisy data and the problem of kernel construction.

\appendices

\section*{Appendix: Proofs}
\subsection{Posterior variance with singular Gramian}
\label{app:gp_proof}
Let $\mathbf{y}_j \in \mathbb{R}^N$ be the $j$-th column of the $\mathsf{y}$-data matrix $\mathbf{y} \in \mathbb{R}^{N \times n}$ from \eqref{eq:reg_kernel_sol}, i.e., $i$-th entry of $\mathbf{y}_j$ is the observation of $f_j(\mathsf{z}_i)$. In the GP context \cite{Rasmussen2005}, we introduce the random variables $\Tilde{Y}$ and $Y$ as the estimates of $f_j(\mathsf{z})$ and $\vstack(f_j(\mathsf{z}_1), \ldots, f_j(\mathsf{z}_N)) \in \mathbb{R}^N$, respectively. Their prior distribution is the joint Gaussian
\[
    \vstack(Y, \Tilde{Y}) \sim \mathcal{N}\left(0, \vstack\left(\mqty[\mathbf{K} & \mathbf{k}(\mathsf{z})], \mqty[\mathbf{k}(\mathsf{z})^\top & \kappa(\mathsf{z}, \mathsf{z})]\right) \right).
\]
We write the joint covariance as a matrix $A$ with blocks $A_{11} \doteq \mathbf{K}$, $A_{12} \doteq \mathbf{k}(\mathsf{z})$, and $A_{22} \doteq \kappa(\mathsf{z}, \mathsf{z})$. Since $A$ is symmetric positive semi-definite, there exists a (not necessarily unique) Cholesky decomposition $A = L L^\top$. We decompose $L$ into blocks $L_{11}, L_{12}$, and $L_{22}$. Rewriting $A = L L^\top$ yields $A_{11} = L_{11} L_{11}^\top$, $A_{12} = L_{11} L_{21}^\top$, and $A_{22} = L_{21} L_{21}^\top + L_{22} L_{22}^\top$. Next, we reformulate $Y$ and $\Tilde{Y}$ from above via $Y = L_{11} \xi_1$ and $\Tilde{Y} = L_{21} \xi_1 + L_{22} \xi_2$, where $\xi_1 \sim \mathcal{N}(0, I)$ and $\xi_2 \sim \mathcal{N}(0, 1)$ are independent. From the former equation, we extract $\xi_1 = L_{11}^\dagger Y + (I - L_{11}^\dagger L_{11}) \hat{\xi}$ with $\hat{\xi} \sim \mathcal{N}(0, I)$ independent of $\xi_1, \xi_2$ and $(I - L_{11}^\dagger L_{11}) \hat{\xi}$ living in the kernel of $L_{11}$. Assuming exact measurements, i.e., $f(\mathsf{z}_i) = \mathsf{y}_i, \forall \, i \in \mathbb{I}_{[1,N]}$, we have that $Y$ attains a deterministic value $Y \doteq \mathbf{y}_j$. This substitution yields
\begin{equation} \label{eq:gp_posterior}
     \Tilde{Y} = L_{21} L_{11}^\dagger \mathbf{y}_j + L_{21} (I - L_{11}^\dagger L_{11}) \hat{\xi} + L_{22} \xi_2.
\end{equation}

\begin{proposition}
    The mean of $\Tilde{Y}$ in \eqref{eq:gp_posterior} is $\mathbf{k}(\mathsf{z})^\top \mathbf{K}^\dagger \mathbf{y}_j$, whereas its variance is $\sigma^2(\mathsf{z}) = \kappa(\mathsf{z}, \mathsf{z}) - \mathbf{k}(\mathsf{z})^\top \mathbf{K}^\dagger \mathbf{k}(\mathsf{z})$. \hfill {\small{$\Box$}}
\end{proposition}
\begin{proof}
    The mean of $\Tilde{Y}$ is $L_{21} L_{11}^\dagger \mathbf{y}$, while the variance of $\Tilde{Y}$ is $L_{21} (I - L_{11}^\dagger L_{11}) (I - L_{11}^\dagger L_{11})^\top L_{21}^\top + L_{22} L_{22}^\top$. Utilizing $L_{11}^\top (L_{11} L_{11}^\top)^\dagger = L_{11}^\top (L_{11}^\top)^\dagger L_{11}^\dagger = L_{11}^\dagger$, we get $\mathbf{k}(\mathsf{z})^\top \mathbf{K}^\dagger = A_{12}^\top A_{11}^\dagger = L_{21} L_{11}^\top (L_{11} L_{11}^\top)^\dagger = L_{21} L_{11}^\dagger$. Hence, for the mean of $\Tilde{Y}$ we have $L_{21} L_{11}^\dagger \mathbf{y}_j = \mathbf{k}(\mathsf{z})^\top \mathbf{K}^\dagger \mathbf{y}_j$, cf. \eqref{eq:reg_kernel_sol}. Similarly,
    \begin{equation*}
        \begin{split}
            &\sigma^2(\mathsf{z}) = \kappa(\mathsf{z}, \mathsf{z}) - \mathbf{k}(\mathsf{z})^\top \mathbf{K}^\dagger \mathbf{k}(\mathsf{z}) = A_{22} - A_{12}^\top A_{11}^\dagger A_{12} = \\
            &= L_{21} (I - L_{11}^\dagger L_{11}) L_{21}^\top + L_{22} L_{22}^\top.
        \end{split}\vspace*{-1mm}
    \end{equation*}
    With $L_{11}^\dagger L_{11} L_{11}^\top = L_{11}^\top$ 
    for the variance of $\Tilde{Y}$ in \eqref{eq:gp_posterior}, one gets
    $L_{21} (I - L_{11}^\dagger L_{11}) (I - L_{11}^\dagger L_{11})^\top L_{21}^\top + L_{22} L_{22}^\top = L_{21} (I - L_{11}^\dagger L_{11}) L_{21}^\top + L_{22} L_{22}^\top$. This completes the proof.
\end{proof}

\subsection{Proof of Lemma~\ref{lem:absolute_error}} \label{app:error_proof}
For each $j \in \mathbb{I}_{[1,n]}$, we have that $f_j^\star(\mathsf{z}) \doteq \mathbf{y}_j^\top \mathbf{K}^\dagger \mathbf{k}(\mathsf{z})$, where $\mathbf{y}_j \doteq \vstack(\mathbf{y}_{1j}, \ldots, \mathbf{y}_{Nj})$ is the $j$-th column of the $\mathsf{y}$-data $\mathbf{y} \in \mathbb{R}^{N \times n}$ from \eqref{eq:reg_kernel_sol}. Set $q \doteq \mathbf{K}^\dagger \mathbf{k}(\mathsf{z}) = \left[q_i\right]_{i=1}^N \in \mathbb{R}^N$. Assuming noise-free data, i.e., $\mathbf{y}_{ij} = f_j(\mathsf{z}_i), \forall \, i \in \mathbb{I}_{[1,N]}$, we get $f_j^\star(\mathsf{z}) = \sum_{i=1}^N \mathbf{y}_{ij} q_i = \sum_{i=1}^N f_j(\mathsf{z}_i) q_i = \sum_{i=1}^N \expval{f_j, \kappa(\cdot, \mathsf{z}_i)}_\mathcal{H} q_i$. Then,
\begin{equation*}
    \begin{split}
        \abs{f_j(\mathsf{z}) - f_j^\star(\mathsf{z})}^2 &= \abs{\expval{f_j, \kappa(\cdot, \mathsf{z})}_\mathcal{H} - \textstyle\sum\nolimits_{i=1}^N \expval{f_j, \kappa(\cdot, \mathsf{z}_i)}_\mathcal{H} q_i}^2 \\
        &\leq \norm{f_j}_\mathcal{H}^2 \norm{\kappa(\cdot, \mathsf{z}) - \textstyle\sum\nolimits_{i=1}^N \kappa(\cdot, \mathsf{z}_i)q_i}^2_\mathcal{H}.
    \end{split}
\end{equation*}
Expansion of the second norm in the product above yields:
\begin{equation*}
    \begin{split}
        &\norm{\kappa(\cdot, \mathsf{z})}^2_\mathcal{H} - 2 \textstyle\sum\nolimits_{i=1}^N \expval{\kappa(\cdot, \mathsf{z}), \kappa(\cdot, \mathsf{z}_i)}_\mathcal{H} q_i \, + \\
        &+ \textstyle\sum\nolimits_{i,j=1}^N \expval{\kappa(\cdot, \mathsf{z}_i), \kappa(\cdot, \mathsf{z}_j)}_\mathcal{H} q_i q_j =\kappa(\mathsf{z}, \mathsf{z}) - 2 \mathbf{k}(\mathsf{z})^\top q \, + \\
        &+ q^\top \mathbf{K} q = \kappa(\mathsf{z}, \mathsf{z}) - 2 \mathbf{k}(\mathsf{z})^\top \mathbf{K}^\dagger \mathbf{k}(\mathsf{z}) + \mathbf{k}(\mathsf{z})^\top \mathbf{K}^\dagger \mathbf{K} \mathbf{K}^\dagger \mathbf{k}(\mathsf{z}) = \\
        &= \kappa(\mathsf{z}, \mathsf{z}) - \mathbf{k}(\mathsf{z})^\top \mathbf{K}^\dagger \mathbf{k}(\mathsf{z}) = \sigma^2(\mathsf{z}).
    \end{split}
\end{equation*}
Therefore, $\abs{f_j(\mathsf{z}) - f_j^\star(\mathsf{z})}^2 \leq \norm{f_j}_\mathcal{H}^2 \sigma^2(\mathsf{z})$. We arrive at \eqref{eq:abs_error_bound} by summing up this inequality over all indices $j \in \mathbb{I}_{[1,n]}$. \hfill {\small{$\blacksquare$}}

\section*{Acknowledgements}
The authors thank the reviewers for very helpful comments.

\bibliography{main}

\end{document}